\documentclass[conference,letterpaper]{IEEEtran}

\renewcommand{\baselinestretch}{0.97}

\usepackage[utf8]{inputenc} 
\usepackage[T1]{fontenc}
\usepackage{ifthen}
\usepackage{cite}
\usepackage[cmex10]{amsmath} 
\usepackage{framed,dsfont,amsfonts}

\newtheorem{theorem}{Theorem}
\newtheorem{corollary}[theorem]{Corollary}

\newtheorem{lemma}[theorem]{Lemma}
\newenvironment{proof}[1][Proof]{\noindent\textbf{#1.} }{\ \rule{0.5em}{0.5em}}

\newcommand{\MyParagraph}[1]{\medskip \noindent {\bf #1}}	
		
\begin{document}

\title{An Information-Theoretic Proof of the Streaming Switching Lemma for Symmetric Encryption}

\author{\IEEEauthorblockN{Ido Shahaf}
\IEEEauthorblockA{Hebrew University of Jerusalem\\ ido.shahaf@cs.huji.ac.il}
\and
\IEEEauthorblockN{Or Ordentlich}
\IEEEauthorblockA{Hebrew University of Jerusalem\\ or.ordentlich@mail.huji.ac.il}
\and
\IEEEauthorblockN{Gil Segev}
\IEEEauthorblockA{Hebrew University of Jerusalem\\ segev@cs.huji.ac.il}
\IEEEoverridecommandlockouts
\IEEEcompsocitemizethanks{
\IEEEcompsocthanksitem
Ido Shahaf and Gil Segev were supported by the European Union's Horizon 2020 Framework Program (H2020) via an ERC Grant (Grant No.\ 714253), and  Ido Shahaf was additionally supported by the Clore Israel Foundation via the Clore Scholars Programme. Or Ordentlich was supported by the ISF under Grant 1791/17.
}
}

\maketitle

\begin{abstract} 
	Motivated by a fundamental paradigm in cryptography, we consider a recent variant of the classic problem of bounding the distinguishing advantage between  a random function and a random permutation. Specifically, we consider the problem of deciding whether a sequence of $q$ values was sampled uniformly  {\em with} or {\em without} replacement from $[N]$, where the decision is made by a streaming algorithm restricted to using at most $s$ bits of internal memory.
	In this work, the distinguishing advantage of such an algorithm is measured by the KL divergence between the distributions of its output as induced under the two cases.
	We show that for any $s=\Omega(\log N)$ the distinguishing advantage is upper bounded by $O(q \cdot s / N)$, and even by $O(q \cdot s / N \log N)$ when $q \leq N^{1 - \epsilon}$ for any constant $\epsilon > 0$ where it is nearly tight with respect to the KL divergence.
\end{abstract}

\section{Introduction}\label{Sec:Introduction}

A fundamental paradigm in the design and analysis of symmetric encryption schemes is the following two-step process: (1) Design a symmetric encryption scheme assuming the availability of a uniformly-random \textit{permutation}; (2) Analyze the security of the scheme assuming that the permutation is switched to a uniformly-random \textit{function}. 

Step (1) relies on the widely-believed existence of pseudorandom permutations (see, for example, \cite{Goldreich01,KatzL14}), which are efficiently-computable and efficiently-invertible keyed permutations $\{ \Pi_{\sf k} \}_{{\sf k} \in \mathcal{K}}$ over $\{0,1\}^n$ that are computationally indistinguishable from a uniformly-random permutation in a standard cryptographic sense, where $\mathcal{K}$ is the set of all possible keys ${\sf k}$. Pseudorandom permutations are realized via a variety of known practical constructions, such as the well-studied and standardized Advanced Encryption Standard ${\sf AES} = \{ {\sf AES}_{\sf k} \}_{{\sf k} \in \mathcal{K}}$ for which $n = 128$.

Step (2) relies on the fact that a uniformly-random function can serve as a perfectly-secure one-time pad for the encryption of an exponentially-large number of messages. For example, assuming that two parties secretly share a uniformly-random permutation $\Pi$ over $\{0,1\}^n$ (this would correspond to actually sharing a key ${\sf k}$ for a pseudorandom permutation), they can use the widely-deployed counter mode for the encryption of multiple messages, and encrypt their $i$th message $m_i \in \{0,1\}^n$ as the pair $c_i = (i, \Pi(i) \oplus m_i)$. Modifying the scheme by replacing its random permutation $\Pi$ with a random function $F : \{0,1\}^n \rightarrow \{0,1\}^n$ enables to argue that an attacker observing a sequence of $q \leq 2^n = N$ ciphertexts $c_1, \ldots, c_q$ obtains no information on their corresponding messages $m_1, \ldots, m_q$. Note, however, that these ciphertexts result from the modified scheme that uses the function $F$, and not from the original one that uses the permutation $\Pi$. Thus, it must be argued that the security of the modified scheme provides a meaningful guarantee for the security of the original one. 

\MyParagraph{The switching lemma.} The security of the modified scheme and that of the original scheme are tied together via a simple argument, commonly referred to as the ``switching lemma''. This lemma captures the  advantage of distinguishing between a random permutation and a random function. For an algorithm (an attacker) that observes $q$ ciphertexts, this translates to upper bounding its advantage in distinguishing a sequence  of $q$ values that are sampled uniformly {\em with replacement} from $\{0,1\}^n = [N]$ (corresponding to the values $F(1), \ldots, F(q)$ in the case of a random function $F$) from a sequence of $q$  values that are sampled  uniformly {\em without replacement} from $[N]$ (corresponding to the values $\Pi(1), \ldots, \Pi(q)$ in the case of a random permutation $\Pi$). The distinguishing advantage of such an algorithm is defined by the dissimilarity between the distribution of its output as induced under the two cases.  Note that the total variation distance between these two distributions is  $\Theta(q^2 / N)$, and this serves as a tight bound on the distinguishing advantage when no restrictions are placed on the distinguisher. 

This implies, in particular, that encryption in the widely-deployed counter mode cannot be used when the number $q$ is approaching $ \sqrt{N}$ messages. In fact, the switching lemma is applicable, and places rather similar  bounds on the number of encrypted messages, not only for symmetric encryption in the above-described counter mode but also for other fundamental modes of encryption. We refer the reader to the work of Jaeger and Tessaro \cite{JaegerT19} for an in-depth discussion of the cryptographic applications on the switching lemma.

\MyParagraph{The streaming switching lemma.} As discussed above, the bound provided by the switching lemma is tight when no restrictions are placed on the distinguisher. Specifically, the following simple algorithm achieves the bound: When given a sequence of $q$ values as input, the algorithm outputs $1$ if there is some value that appears more than once (i.e., if a ``collision'' exists), and outputs $0$ otherwise. Note that when given a sequence of values that are sampled uniformly with replacement this algorithm outputs $1$ with probability $\Theta(q^2 / N)$, and when given a sequence of values that are sampled uniformly  without replacement  this algorithm always outputs $0$. However, a significant drawback of this algorithm is that it needs an internal memory of size  $q \cdot \log N$ bits for storing the entire sequence in order to identify whether or not a collision exists.

This observation motivated Jaeger and Tessaro \cite{JaegerT19} to refine the framework of the switching lemma by restricting the amount of internal memory used by the distinguisher. That is, they analyzed the advantage of distinguishing the above two distributions where: (1) the $q$ values are provided one by one in a streaming manner, and (2) the internal memory of the distinguisher is restricted to at most $s$ bits. The most interesting regime is where there is a noticeable gap between $s$ and $q \cdot \log N$, which is motivated by the fact that large amounts of data cannot always be stored in their entirety.

\MyParagraph{Known bounds.} Jaeger and Tessaro proved a conditional upper bound on the distinguishing advantage of any streaming algorithm $\mathcal{A}$ that uses at most $s$ bits of internal memory. Specifically, they introduced a combinatorial conjecture regarding certain hypergraphs, and showed that based on their conjecture the advantage of any such distinguisher is at most $O(q \cdot s / N)$, when  measured as the KL divergence between the output distributions of the memory-bounded streaming algorithm $\mathcal{A}$ under the two cases. Applying Pinker's inequality, this implies an upper bound of $O(\sqrt{q \cdot s / N})$ when measured via the total variation distance, which is more standard for cryptographic applications.

In a follow-up work, Dinur \cite{Dinur19} proved an unconditional upper bound of $O((q \log q) \cdot s/ N)	$ on the distinguishing advantage of any such algorithm, when measured as the total variation distance between the output distributions of the memory-bounded streaming algorithm under the two cases. Note that this should be compared to the upper bound $O(\sqrt{q \cdot s / N})$ on the total variation distance obtained by applying Pinsker's inequality to the result of Jaeger and Tessaro.

Dinur's result is based on reducing the task of distinguishing between these two distributions via a memory-bounded algorithm to constructing communication-efficient protocols for the two-party set-disjointness problem. Three decades of extensive research on the communication complexity of this canonical problem (e.g., \cite{Bar-YossefJKS04, KalyanasundaramS92, Razborov92}) have recently led to new lower bounds \cite{Goos016} on which Dinur relied via his reduction.

\MyParagraph{Our contributions.} We present an information-theoretic and unconditional  proof  showing that the distinguishing advantage of any streaming algorithm that uses at most $s = \Omega(\log N)$ bits of internal memory  is at most $O(q \cdot s / N)$,  measured via KL divergence as in the work of Jaeger and Tessaro \cite{JaegerT19}. When $q \leq N^{1 - \epsilon}$ for any constant $\epsilon > 0$, we obtain an improved upper bound of $O(q \cdot s / N \log N)$ which is asymptotically tight with respect to the KL divergence.

Moreover, we prove our results within a more refined framework that considers the {\em accumulated} memory usage of streaming algorithms throughout their computation, and not only their worst-case memory usage. This shows that any non-negligible advantage must be obtained by using a substantial amount of internal memory on average throughout the computation, and not only in the worst case.

\section{Setup and Main Results}\label{Sec:Setup}

\MyParagraph{Notation.} All logarithms in this paper are to the natural base unless denoted otherwise in a subscript. For two probability distributions $P_X$ and $Q_X$ on a common discrete alphabet $\mathcal{X}$, where $P_X$ is absolutely continuous with respect to $Q_X$, the KL-divergence is defined as $D_{\mathsf{KL}}(P_X\|Q_X)=\sum_{x\in\mathcal{X}} P_X(x)\log\frac{P_X(x)}{Q_X(x)}$. For probability distributions $P_{XY}=P_X P_{Y|X}$ and $Q_{XY}=Q_X Q_{Y|X}$ on a common discrete alphabet $\mathcal{X}\times\mathcal{Y}$, where $P_{XY}$ is absolutely continuous with respect to $Q_{XY}$,   we further define the conditional divergence as $D_{\mathsf{KL}}(P_{Y|X}\|Q_{Y|X}|P_X)=\sum_{x\in\mathcal{X}} P_{X}(x)D_{\mathsf{KL}}(P_{Y|X=x}\|Q_{Y|X=x})$. The mutual information between $X$ and $Y$ with respect to the probability distribution $P_{XY}$ is $I(X;Y)=D_{\mathsf{KL}}(P_{Y|X}\|P_Y|P_X)=H(Y)-H(Y|X)$, where $H(Y)=\sum_{y\in\mathcal{Y}} P_{Y}(y)\log\frac{1}{P_Y(y)}$ and $H(Y|X)=\sum_{x\in\mathcal{X},y\in\mathcal{Y}} P_{XY}(x,y)\log\frac{1}{P_{Y|X}(y|x)}$.

\MyParagraph{Setup.} For stating our results we briefly describe the notion of memory-bounded streaming indistinguishability, introduced by Jaeger and Tessaro \cite{JaegerT19}, as well as our refinement that considers accumulated memory usage. For an algorithm $\mathcal{A}$ and a sequence $x=(x_1,\dots,x_q) \in [N]^q$, $q<N$, the streaming computation of $\mathcal{A}$ on $x$ is defined via the following process:
\begin{itemize}
	\item Set $\sigma_0=\bot$, where $\bot$ is the empty string.
	\item For $i=1,\dots,q$:
	\begin{itemize}
		\item Let $\sigma_i\leftarrow\mathcal{A}(i,\sigma_{i-1},x_i)$.
	\end{itemize}
	\item Output $\sigma_q$.
\end{itemize}

We abuse notation and denote the output of this computation by $\mathcal{A}(x)$. Following Jaeger and Tessaro, we say that an algorithm $\mathcal{A}$ is $s$-memory-bounded if for every input $x \in [N]^q$ and for every $i\in[q]$ it holds that $|\sigma_i|= s$, where $|\sigma_i|$ is the bit length of the internal state $\sigma_i$. For our purpose of considering accumulated memory usage, we naturally extend this notion to that of an $(s_1, \ldots, s_q)$-memory-bounded algorithm, where for every input $x \in [N]^q$ and for every $i\in[q]$ it holds that $|\sigma_i|= s_i$. From this point, and without loss of generality, we assume that for any $(s_1, \ldots, s_q)$-memory-bounded algorithm it holds that $s_{i+1}\le s_i+ \lceil \log_2 N \rceil$ for all $i\in[q-1]$ and that it holds that $s_1\le \lceil \log_2 N \rceil$.\footnote{For any sequence $s_1,\dots,s_q$, we may recursively define $s'_1,\dots,s'_q$ by $s'_1=\min\{s_1, \lceil \log_2 N \rceil\}$ and $s'_{i+1}=\min\{s_{i+1}, s'_i + \lceil \log_2 N \rceil\}$. Then, any $(s_1,\dots,s_q)$-memory-bounded algorithm $\mathcal{A}$ with internal states $\sigma_1, \ldots, \sigma_q$ can be transformed into an $(s'_1,\dots,s'_q)$-memory-bounded algorithm $\mathcal{A}'$ with internal states $\sigma'_1, \ldots, \sigma'_q$ by defining $\sigma'_{i+1} = \sigma_{i+1}$ if $s'_{i+1}=s_{i+1}$ and defining $\sigma'_{i+1} = (\sigma'_{i},x_{i+1})$ otherwise, where $(x_1, \ldots, x_q)$ is the input sequence (i.e., $x_{i+1}$ can always be stored explicitly together with the previous state $\sigma_i$ instead of updating the state to $\sigma_{i+1}$). Note that $\mathcal{A}'$ perfectly simulates the execution of $\mathcal{A}$ for any input, and thus achieve the same distinguishing advantage.}

From this point on we let $Q$ and $P$ denote the probability distributions on $[N]^q$ corresponding to sampling the sequence $X=(X_1,\ldots,X_q)$
uniformly {\em with} and {\em without} replacement, respectively, from $[N]$. Namely, under $Q$ we have that $X_i\stackrel{\text{i.i.d.}}{\sim}\mathrm{Uniform}([N])$ for $i\in[q]$, whereas under $P$ we have that $X_1,\ldots,X_q$ are the first $q$ entries of a uniform random permutation on $[N]$. The distribution of the algorithm's output under $Q$ (respectively $P$) is denoted by $Q_{\mathcal{A}}$ (respectively $P_{\mathcal{A}}$).

\MyParagraph{Main results.} The following theorem states our main result, upper bounding the distinguishing advantage of any memory-bounded streaming algorithm, when measured via KL divergence:

\medskip
\begin{theorem}\label{Thm:Main}
	For any $N\geq 1$, $q = o(N)$ and $s_1, \ldots, s_q$ such that $0\le s_i = O(N)$ for all $i \in [q]$, and for any $(s_1, \ldots, s_q)$-memory-bounded algorithm $\mathcal{A}$ it holds that
	\[ D_{\mathsf{KL}}(P_{\mathcal{A}}||Q_{\mathcal{A}})\le (1 + o(1)) \cdot \frac{\sum_{i=1}^{q-1} s_i + q\cdot \log_2 N}{N \log_2(N/q)}  \; . \]%
\end{theorem}
\medskip

In particular, when $q \leq N^{1 - \epsilon}$ for any constant $\epsilon > 0$, then also $s_i\le O(q\cdot \log N)\le O(N)$, and we obtain the following corollary:
\medskip
\begin{corollary}\label{Cor:Main}
	For any $N\geq 1$, constant $\epsilon > 0$, $q \leq N^{1 - \epsilon}$ and $s_1, \ldots, s_q$ such that $s_i \ge 0$ for all $i \in [q]$, and for any $(s_1, \ldots, s_q)$-memory-bounded algorithm $\mathcal{A}$ it holds that
	\[ D_{\mathsf{KL}}(P_{\mathcal{A}}||Q_{\mathcal{A}})\le (1+o(1)) \cdot \frac{\sum_{i=1}^{q-1} s_i+q\cdot\log_2N}{\epsilon\cdot N \log_2 N}   \; . \]%
\end{corollary}
\medskip

Finally, for this range of parameters we observe that our bound is nearly tight: 
\medskip
\begin{theorem}\label{Thm:Tight}
	For any $N\geq 1$, and $s_1, \ldots, s_q$ such that $s_i \ge 1$ for all $i \in [q]$, there exists an $(s_1, \ldots, s_q)$-memory-bounded algorithm $\mathcal{A}$ for which
	\[ D_{\mathsf{KL}}(P_{\mathcal{A}}||Q_{\mathcal{A}}) \geq \frac{\sum_{i=1}^{q-1} s_i-q\cdot(\log_2 N+1)}{N \log_2 N}   \; .\]%
\end{theorem}

\section{Proof of Theorem \ref{Thm:Main}}\label{Sec:Proof}

Our proof is based on an induction argument showing that $D_{\mathsf{KL}}(P_{\mathcal{A}}\|Q_{\mathcal{A}})\leq \sum_{i=1}^q I(X_i;\Sigma_{i-1})$, where the mutual information is computed with respect to $P$, and $\Sigma_{i-1}$ is the state of the internal memory at step $i-1$ of the computation. Then, we leverage the fact that $\Sigma_{i-1}-(X_1,\ldots,X_{i-1})-X_i$ form a Markov chain in this order, and that $I(\Sigma_{i-1};X_1,\ldots,X_{q-1})\leq s_{i-1}\log{2}$ due the the memory constraints, in order to derive an information bottleneck \cite{tpb99} upper bound on $I(X_i;\Sigma_{i-1})$.

\subsection{An Induction Argument}

We prove the following lemma which is similar to a lemma proved by Jaeger and Tessaro \cite{JaegerT19}.

\medskip
\begin{lemma}\label{Lem:JT}
	Let $P_{X}$ and $Q_{X}$ be two distributions on $\mathcal{X}^q$, where the induced marginals satisfy $P_{X_i}=Q_{X_i}$ for all $i=1,\ldots,n$, and in addition $Q_{X}=\prod_{i=1}^q Q_{X_i}$ (i.e., under the distribution $Q_X$ the random variables $X_1,\dots,X_q$ are independent, where each $X_i$ is distributed according to the distribution $Q_{X_i}$). For a streaming computation performed by the algorithm $\mathcal{A}$, let $\Sigma_i=\Sigma_i(X_1,\ldots,X_i)$  be the random variable corresponding to the state $\sigma_i$ produced in the $i$th step of the computation. Then
	\begin{align}
	D_{\mathsf{KL}}(P_{\mathcal{A}}||Q_{\mathcal{A}})&\le \sum_{i=1}^q D_{\mathsf{KL}}(P_{X_i|\Sigma_{i-1}}\|P_{X_i}|P_{\Sigma_{i-1}}) \nonumber\\
	&=\sum_{i=1}^{q}I(X_i;\Sigma_{i-1}) \nonumber ,
	\end{align}
	where the mutual information is computed with respect to the joint distribution $P_{X_i\Sigma_{i-1}}$, induced by $P_X$.
\end{lemma}

\begin{proof}
	By definition of $\Sigma_i$ we have that $D_{\mathsf{KL}}(P_{\mathcal{A}}\|Q_{\mathcal{A}})=D_{\mathsf{KL}}(P_{\Sigma_q}\|Q_{\Sigma_q})$. Moreover, since $\Sigma_q$ is obtained by processing $(\Sigma_{q-1},X_q)$, the data processing inequality yields 
	\[
	D_{\mathsf{KL}}(P_{\Sigma_q}\|Q_{\Sigma_q}) \le D_{\mathsf{KL}}(P_{X_q\Sigma_{q-1}}\|Q_{X_q\Sigma_{q-1}}) \; .
	\]
	Applying the chain rule, yields
	\begin{align}
	& D_{\mathsf{KL}}(P_{X_q\Sigma_{q-1}}\|Q_{X_q\Sigma_{q-1}})\nonumber \\ 
	&  = D_{\mathsf{KL}}(P_{\Sigma_{q-1}}\|Q_{\Sigma_{q-1}}) + D_{\mathsf{KL}}(P_{X_q|\Sigma_{q-1}}\|Q_{X_q|\Sigma_{q-1}}|P_{\Sigma_{q-1}}) \nonumber \\
	&  = D_{\mathsf{KL}}(P_{\Sigma_{q-1}}\|Q_{\Sigma_{q-1}}) + D_{\mathsf{KL}}(P_{X_q|\Sigma_{q-1}}\|Q_{X_q}|P_{\Sigma_{q-1}}) \label{Eq:Pre1} \\
	& = D_{\mathsf{KL}}(P_{\Sigma_{q-1}}\|Q_{\Sigma_{q-1}}) + D_{\mathsf{KL}}(P_{X_q|\Sigma_{q-1}}\|P_{X_q}|P_{\Sigma_{q-1}})\label{Eq:Pre2},
	\end{align}%
	where~\eqref{Eq:Pre1} follows from the fact that $Q_{X}$ is memoryless such that under this distribution $X_q$ is statistically independent of  $\Sigma_{q-1}=\Sigma_{q-1}(X_1,\ldots,X_{q-1})$, and~\eqref{Eq:Pre2} follows from the assumption that $P_{X_q}= Q_{X_q}$. 	Thus, by induction we obtain that $D_{\mathsf{KL}}(P_{\Sigma_q}\|Q_{\Sigma_q}) \le D_{\mathsf{KL}}(P_{\Sigma_0}\|Q_{\Sigma_0}) +\sum_{i=1}^{q}D_{\mathsf{KL}}(P_{X_i|\Sigma_{i-1}}\|P_{X_i}|P_{\Sigma_{i-1}})$. Recalling that $P_{\Sigma_0}=Q_{\Sigma_0}$ and that $D_{\mathsf{KL}}(P_{X_i|\Sigma_{i-1}}\|P_{X_i}|P_{\Sigma_{i-1}})=I(X_i;\Sigma_{i-1})$, our claim follows.
\end{proof}

\subsection{An Information-Bottleneck Argument}

We make use of the following functions:
\begin{itemize}
	\item For $x\in[0,1]$ the binary entropy function (with respect to the natural basis) is
	\[
	h_2(x)=-x\log(x)-(1-x)\log(1-x) \; ,
	\]
	and we let $h_2^{-1}$ be its inverse restricted to $[0,1/2]$.
	\item For $y\le 1$ we let $f(y)=-(1-y)\log (1-y)$.
	\item For $t\in[0,\log2]$ we let
	\[
	\varphi(t)= f(h_2^{-1}(t))= -(1-h_2^{-1}(t))\log\left(1-h_2^{-1}(t)\right) \; ,
	\]
	and for $t<0$ we let $\varphi(t)=0$.
\end{itemize}
We claim that $f$ is non-decreasing over $[0,1/2]$, that $\varphi$ is non-decreasing and convex, and that for every $t\in[0,1]$ it holds that $f(t)\geq \varphi(h_2(t))$. We defer the proofs to Section \ref{Sec:Prop}.
We state and prove our main technical lemma.

\medskip
\begin{lemma}\label{Lem:Main}
	Let $0\le i< N$ be integers, let $X=(X_1,\ldots,\allowbreak X_{i+1})$ be the random process of sampling $i+1$ elements of $[N]$ uniformly without replacement. Denote $V=(X_1,\ldots,X_i)$, $W=X_{i+1}$, and let $\Gamma$ be a random variable such that $\Gamma-V-W$ form a Markov chain in this order.
	Then, it holds that
	\[
	I(W;\Gamma) \le \log\frac{N}{N-i}-\frac{N}{N-i}\cdot\varphi\left(\frac{\log\binom{N}{i}}{N}-\frac{I(V;\Gamma)}{N}\right)
	\]
\end{lemma}

\begin{proof}	
	We first note that
	\begin{align}
	I(W;\Gamma)&=H(W)-H(W|\Gamma)=\log N-H(W|\Gamma) \; , \label{eq:equivalentdef}
	\end{align}	
	and that
	\begin{align}
	I(V;\Gamma)&=H(V)-H(V|\Gamma)\nonumber\\
	&=\log\frac{N!}{(N-i)!}-H(V|\Gamma) \; . \label{eq:equivalentdef2}
	\end{align}
	Consequently, we derive a lower bound on $H(W|\Gamma)$ in terms of $H(V|\Gamma)$. To that end, we first compute the distribution $\Pr(W=j|\Gamma=\gamma)$ for $j\in[N]$ and $\gamma\in\mathsf{supp}(\Gamma)$. We have
	\begin{align*}
	& \Pr(W=j|\Gamma=\gamma) \\
	&\qquad = \Pr(W=j,j\notin V|\Gamma=\gamma) \\
	&\qquad = \Pr(j\notin V|\Gamma=\gamma)\Pr(W=j| j\notin V,\Gamma=\gamma) \\
	&\qquad = \Pr(j\notin V|\Gamma=\gamma)\Pr(W=j| j\notin V) \\
	&\qquad = \frac{1-\Pr(j\in V|\Gamma=\gamma)}{N-i}.
	\end{align*}
	It follows that
	\begin{align}
	& H(W|\Gamma=\gamma) \nonumber\\
	& \qquad =\sum_{j=1}^N \Pr(W=j|\Gamma=\gamma)\log\frac{1}{\Pr(W=j|\Gamma=\gamma)}\nonumber\\
	& \qquad =\sum_{j=1}^N \frac{1-\Pr(j\in V|\Gamma=\gamma)}{N-i}\log \frac{N-i}{1-\Pr(j\in V|\Gamma=\gamma)}\nonumber\\
	& \qquad =\log(N-i)+\frac{1}{N-i} \sum_{j=1}^N f\left(\Pr(j\in V|\Gamma=\gamma)\right)\nonumber\\
	& \qquad \geq \log(N-i)+\frac{1}{N-i} \sum_{j=1}^N\varphi\left(h_2\left(\Pr(j\in V|\Gamma=\gamma)\right)\right) \; . \label{eq:condentbound}
	\end{align}
	Defining the random variables $A_j=\mathds{1}_{\{j\in V\}}$, we further write 
	\begin{align}
	& \sum_{j=1}^N\varphi\left(h_2\left(\Pr(j\in V|\Gamma=\gamma)\right)\right)\nonumber\\
	& \qquad =\sum_{j=1}^N\varphi\left(H(A_j|\Gamma=\gamma)\right)\nonumber\\
	& \qquad \geq N\varphi\left(\frac{1}{N}\sum_{j=1}^NH(A_j|\Gamma=\gamma)\right),\label{eq:probbound}
	\end{align}
	where in the last step we used the convexity of $\varphi$.
	Next, using the fact that conditioning reduces entropy, we note that
	\begin{align*}
	\sum_{j=1}^N H(A_j|\Gamma=\gamma)&\geq \sum_{j=1}^N H(A_j|A_1,\dots,A_{j-1},\Gamma=\gamma) \\
	&=H(A_1,\ldots,A_N|\Gamma=\gamma).
	\end{align*}
	Note that $A_1,\ldots,A_N$ dictate the elements that belong to $V$. Let $\pi$ be the order at which these elements appear. Together, $(A_1,\ldots,A_N)$ and $\pi$ completely determine $V$, and vice versa. We have that
	\begin{align*}
	& H(A_1,\ldots,A_N|\Gamma=\gamma) \\
	& \qquad = H(A_1,\ldots,A_N,\pi|\Gamma=\gamma)-H(\pi|A_1,\ldots,A_N,\Gamma=\gamma) \\
	& \qquad = H(V|\Gamma=\gamma)-H(\pi|A_1,\ldots,A_N,\Gamma=\gamma) \\
	& \qquad \geq H(V|\Gamma=\gamma)-\log (i!).
	\end{align*}
	Plugging this into~\eqref{eq:probbound} (using the monotonicity of $\varphi$), and then into~\eqref{eq:condentbound}, we obtain
	\begin{align*}
	&H(W|\Gamma=\gamma)\\
	&\qquad\geq \log(N-i)+\frac{N}{N-i} \varphi\left( \frac{H(V|\Gamma=\gamma)-\log (i!)}{N}\right).
	\end{align*}
	Recalling that $H(W|\Gamma)=\mathbb{E}_{\gamma}\left[H(W|\Gamma=\gamma)\right]$ and $H(V|\Gamma)=\mathbb{E}_{\gamma}\left[H(V|\Gamma=\gamma)\right]$, and using the convexity of $\varphi$, we obtain
	\begin{align}
	H(W|\Gamma)\geq \log(N-i)+\frac{N}{N-i}  \varphi\left( \frac{H(V|\Gamma)-\log (i!)}{N}\right) \; , \label{eq:MGL}
	\end{align}
	and the statement follows by plugging~\eqref{eq:MGL} into~\eqref{eq:equivalentdef} using~\eqref{eq:equivalentdef2}.
\end{proof}

Next, we simplify the bound of Lemma \ref{Lem:Main}.
\medskip
\begin{corollary}\label{Cor:est}
	In the setting of Lemma \ref{Lem:Main}, if $i=o(N)$ and $I(V;\Gamma)=O(N)$ then
	\[
	I(W;\Gamma) \le (1+o(1))\cdot\frac{I(V;\Gamma)+\log N}{N\log (N/i)}
	\]
\end{corollary}
\begin{proof}
	To that end, we will use the following well-known estimate (proved using Stirling's approximation, e.g.\ \cite{MacWilliamsS77}):
	\begin{align*}
	& N h_2\left(\frac{i}{N}\right)-\frac{1}{2}\log\left(8 i\left(1-\frac{i}{N}\right)\right) \leq \log{N\choose i} \\
	& \qquad \leq N h_2\left(\frac{i}{N}\right)-\frac{1}{2}\log\left(2\pi i\left(1-\frac{i}{N}\right)\right)
	\end{align*}
	In particular, for large enough $N$ it holds that
	\begin{align*}
	& \log{N\choose i} \ge N h_2\left(\frac{i}{N}\right)-\log N
	\end{align*}
	Let $\alpha=i/N$. Using the monotonicity of $\varphi$, the bound of Lemma \ref{Lem:Main} reads as
	\begin{align}
	-\log(1-\alpha)-\frac{1}{1-\alpha}\varphi\left(h_2(\alpha)-\frac{I(V;\Gamma)+\log N}{N}\right) \; .\label{eq:est1}
	\end{align}
	Let $\beta=(I(V;\Gamma)+\log N)/{N}$. Due to the the convexity of $h_2^{-1}$ it holds that 
	\begin{align*}
	h_2^{-1}\left(h_2(\alpha)-\beta\right)\ge \alpha-g(\alpha)\cdot\beta \; ,
	\end{align*}
	where $g(\alpha)=\left(h_2^{-1}\right)'  (h_2(\alpha))$. Recall that $\varphi(t)=f(h_2^{-1}(t))$ and that $f$ is increasing at $[0,1/2]$, and hence
	\begin{align*}
	\varphi(h_2(\alpha)-\beta)\ge f(\alpha-g(\alpha)\cdot\beta) \; ,
	\end{align*}
	and we can further upper bound \eqref{eq:est1} by
	\begin{align}
	-\log(1-\alpha)-\frac{1}{1-\alpha}f(\alpha-g(\alpha)\cdot\beta) \; .\label{eq:est2}
	\end{align}
	Denoting $\delta=g(\alpha)\cdot\beta/(1-\alpha)$ and recalling the definition of $f$, we further develop \eqref{eq:est2}
	\begin{align}
	& -\log(1-\alpha)+\frac{1-\alpha+g(\alpha)\cdot\beta}{1-\alpha}\log(1-\alpha+g(\alpha)\cdot\beta) \nonumber\\
	& \qquad = -\log(1-\alpha)+\left(1+\delta\right)\log\left((1-\alpha)\left(1+\delta\right)\right) \nonumber\\
	& \qquad = \delta\log(1-\alpha)+\left(1+\delta\right)\log\left(1+\delta\right) \nonumber \\
	& \qquad \le \left(1+\delta\right)\delta \label{eq:est3}  \\
	& \qquad = \frac{1+\delta}{1-\alpha}g(\alpha)\beta \; , \nonumber
	\end{align}
	where in \eqref{eq:est3} we used the inequality $\log(1+\delta)\le \delta$ that holds for every $\delta$.
	Since $\alpha=o(1)$ we can estimate
	\begin{align*}
	g(\alpha)&=\left(h_2^{-1}\right)'  (h_2(\alpha))
	=1/h'_2(\alpha)\\
	&=\frac{1}{\log(1-\alpha)-\log\alpha} \\
	&=\frac{1+o(1)}{\log (N/i)} \; .
	\end{align*}
	Since we assumed that $I(V;\Gamma)=O(N)$, it also holds that $\delta=O(g(\alpha))=o(1)$, and we conclude that
	\[
	I(W;\Gamma) \le (1+o(1))\cdot\frac{I(V;\Gamma)+\log N}{\log (N/i)}
	\]
\end{proof}

\subsection{Application to Our Setup}

Finally, we can derive Theorem \ref{Thm:Main} from Corollary \ref{Cor:est}. Recall that $P$ and $Q$ designate the probability distributions corresponding to sampling $X=(X_1,\ldots,X_q)$ uniformly without and with replacement, respectively, from $[N]$. Thus, $P_{X_i}=Q_{X_i}$ for all $i \in [q]$, and furthermore, $Q$ is a memoryless distribution. Thus, the conditions of Lemma \ref{Lem:JT} hold, and $D_{\mathsf{KL}}(P_{\mathcal{A}}||Q_{\mathcal{A}})\leq \sum_{i=1}^{q}I(X_i;\Sigma_{i-1})$, where the mutual information is with respect to $P$. Now, recalling that $\Sigma_{i-1}-(X_1,\ldots,X_{i-1})-X_i$ forms a Markov chain in this order, and that under $P$ we have that $(X_1,\ldots,\allowbreak X_{i})$ is a random process of sampling $i$ elements of $[N]$ uniformly without replacement, and that $I(\Sigma_i;X_1,\dots,X_i)\le H(\Sigma_i)\le s_i\log 2$ by the constraints on the internal memory, we can apply Corollary \ref{Cor:est} to obtain
\[
I(X_{i+1};\Sigma_i)\le (1+o(1))\cdot\frac{s_i\log 2+\log N}{N\log (N/q)} \; ,
\]
and Theorem \ref{Thm:Main} follows by summing over all $i\in[q-1]$. This settles the proof of Theorem \ref{Thm:Main}.

\section{Proof of Theorem \ref{Thm:Tight}}\label{Sec:Proof2}
Informally, given $s_1,\dots,s_q$ such that $s_i\ge 1$ for all $i\in[q]$ we construct an $(s_1,\dots,s_q)$-memory-bounded algorithm $\mathcal{A}$ that stores a list of values that it saw, where every new value is added to the list if the state size allows it. More formally, with a loss of at most $\log N$ bits per each $s_i$, we may assume that $s_i$ is of the form $1+ k_i\lceil\log_2 N\rceil$ for an integer $k_i$ for all $i\in[q]$. We remind that we assume that $s_{i+1}\le s_i+\lceil\log_2 N\rceil$ for all $i \in [q-1]$ and that $s_1\le \lceil\log_2 N\rceil$, thus it holds that $k_{i+1}\le k_i + 1$ for all $i \in [q-1]$ and it holds that $k_1\le 1$.
We also assume without loss of generality that $s_q=1$ (i.e., the final output of $\mathcal{A}$ is a single bit).
For $i\in[q]$ we define the computation $\mathcal{A}(i,\sigma_{i-1},x_i)$ as follows:
\begin{itemize}
	\item If $i=1$, output $\sigma_1=0$ or $\sigma_1=(0,x_1)$ according to whether $k_1=0$ or $k_1=1$, respectively.
	\item Else, if the first bit of $\sigma_{i-1}$ is $1$, output $\sigma_i=1^{s_i}$.
	\item Else, parse $\sigma_{i-1}=(0,y_1,\dots,y_{k_{i-1}})\in\{0,1\}\times[N]^{k_{i-1}}$.
	\item If $x_i\in\{y_1,\dots,y_{k_{i-1}}\}$, output $\sigma_i=1^{s_i}$.
	\item Else, if $k_i=k_{i-1}+1$, output $\sigma_i=(0,y_1,\dots,y_{k_{i-1}},x_i)$.
	\item Else, output $\sigma_i=(0,y_1,\dots,y_{k_{i}})$.
\end{itemize}
Note that for this algorithm $\mathcal{A}(X)\in\{0,1\}$ and that $P[\mathcal{A}(X)=0]=1$, so it holds that
\begin{align*}
D_{\mathsf{KL}}(P_{\mathcal{A}}||Q_{\mathcal{A}}) &= P[\mathcal{A}(X)=0]\log\left(\frac{P[\mathcal{A}(X)=0]}{Q[\mathcal{A}(X)=0]}\right) \\
&= -\log\left(Q[\mathcal{A}(X)=0]\right) \\
&= -\log\left(\prod_{i=1}^{q-1}\left(1-\frac{k_i}{N}\right)\right) \\
&= -\sum_{i=1}^{q-1}\log\left(1-\frac{k_i}{N}\right) \\
&\ge \sum_{i=1}^{q-1}\frac{k_i}{N} \\
&\ge \frac{\sum_{i=1}^{q-1}s_i-q\cdot(\log_2 N+1)}{N\log_2 N} \; ,
\end{align*}
and this settles the proof of Theorem \ref{Thm:Tight}.

\section{Proofs for the properties of $f$ and $\varphi$}\label{Sec:Prop}
In this section we give proofs for the properties of $f$ and $\varphi$ that we used.

We start by showing that $f$ is increasing over $[0,1/2 ]$. Indeed $f'(x)=\log(1-x)+1$, so $f'(x)>0$ as long as $x<1-1/e\approx0.6$.
Now, we show $\varphi$ is increasing over $[0,\log 2]$.
Recall that $\varphi(t)=f(h_2^{-1}(t))$, and the claim follows from the fact that $h_2^{-1}$ is increasing and $f$ is increasing over $[0,1/2]$.
Next, we show that $\varphi$ is convex by showing that its derivative is increasing. It holds that
\[
\varphi'(t)=\frac{f'(h_2^{-1}(t))}{(h_2)'(h_2^{-1}(t))} \; .
\]
Thus, $\varphi'(t)=p(h_2^{-1}(t))$ where
\[
p(x)=\frac{f'(x)}{(h_2)'(x)}=\frac{\log(1-x)+1}{\log(1-x)-\log x} \; .
\]
Computing the derivative, for $x\in(0,1/2)$ we get
\[
p'(x)=\frac{1-h(x)}{x(1-x)(\log(1-x)-\log x)^2} > 0 \; .
\]
So $p$ is increasing, thus $\varphi'$ is increasing and $\varphi$ is convex as claimed.
Finally, we show that for every $t\in[0,1]$ it holds that $f(t)\le\varphi(h_2(t))$. When $t\in[0,1/2]$ it simply holds that $\varphi(h_2(t))=f(h_2^{-1}(h_2(t)))=f(t)$. When $t\in(1/2,1]$, it holds that $\varphi(h_2(t))=f(h_2^{-1}(h_2(t)))=f(1-t)$, thus we need to show that $f(t)\ge f(1-t)$. Let $g(t)=f(t)-f(1-t)$.
Then, $g''(t)=1/t-1/(1-t)<0$ (when $t\in(1/2,1)$), so $g$ is concave over $[1/2,1]$. Together with the fact that $g(1/2)=g(1)=0$ we get that $g(t)\ge 0$ when $t\in[1/2,1]$.

\pagebreak
\bibliographystyle{IEEEtran}
\bibliography{Main}

\end{document}